\documentclass[12pt,a4paper]{article}
\usepackage{makeidx}
\usepackage{amsmath}
\usepackage{amsfonts, amssymb}
\usepackage{amscd}
\usepackage{graphicx}
\usepackage{longtable}
\newcommand{\EQ}{\begin{equation}}
\newcommand{\EN}{\end{equation}}

\newtheorem{theorem}{Theorem}[section]

\newtheorem{corollary}[theorem]{Corollary}
\newtheorem{proposition}[theorem]{Proposition}
\newtheorem{lemma}[theorem]{Lemma}

\newtheorem{ex}[theorem]{Example}

\newcommand{\pr}{\indent{\em Proof: \ }}
\newcommand{\qed}{\hfill$\triangle$\bigskip}

\newcommand{\F}{{\mathbb{F}}}

\newcommand{\C}{{\cal C}}

\newcommand{\zero}{{\mathbf{0}}}

\newenvironment{proof}{\noindent {\pr}\ }{\qed}
\newcommand{\cH}{{\cal H}}

\newcommand{\supp}{\hbox{supp}}

\newcommand{\APN}{{\textrm{APN }}}
\newcommand{\CCZ}{{\textrm{CCZ}}}
\newcommand{\se}{self-embedding }

\title{Self-embeddings of Hamming Steiner triple systems of small order and \APN
permutations\thanks{This work was partially supported by the
Spanish MICINN under Grants MTM2009-08435 and TIN2010-17358, and
by the Catalan AGAUR under Grant 2009SGR1224. The work of the
second author was supported by the Grants RFBR 10-01-00424-a and
12-01-00631-a. \newline \indent $^1$J. Rif\`{a} and M. Villanueva
are with the Department of Information and Communications
Engineering,
                             Universitat Aut\`{o}noma de Barcelona,
                             08193-Bellaterra, Spain.
                             (email:~\{josep.rifa,
                             merce.villanueva\}@uab.cat) \newline \indent $^2$F. I. Solov'eva is with the Sobolev Institute
of Mathematics and Novosibirsk State University, Novosibirsk,
Russia. (email:~sol@math.nsc.ru)}}

\author{J. Rif\`{a}$^1$,
F. I. Solov'eva$^2$, M. Villanueva$^1$}

\begin{document}

\maketitle

\begin{abstract} The classification, up to isomorphism, of all self-embedding monomial power permutations of  Hamming Steiner triple systems of order $n=2^m-1$ for small $m$, $m \leq 22$, is given. As far as we know, for $m \in \{5,7,11,13,17,19 \}$, all given self-embeddings in  closed surfaces are new. Moreover, they are cyclic for all $m$ and nonorientable at least for all $m \leq 19$.
For any non prime $m$, the nonexistence of such self-embeddings in a closed surface is proven.

The
rotation line spectrum for self-embeddings of  Hamming Steiner triple systems in
 pseudosurfaces with pinch points as an invariant to distinguish \APN permutations or, in general, to classify permutations, is
 also proposed. This classification for \APN monomial power permutations coincides with the \textrm{CCZ}-equivalence,
 at least up to $m\leq 17$.
\end{abstract}

\section{Introduction}
\label{sec:1}

Let  $\F^n$ be the vector space of dimension  $n$ over the binary
field $\F$. The  {\it Hamming distance} between two vectors
$x,y\in \F^n$, denoted by $d(x,y)$, is the number of coordinate
positions in which $x$ and $y$ differ. The {\it Hamming weight} of
$x\in \F^n$, denoted by $w(x)$, is given by $w(x)=d(x,\zero)$,
where $\zero$ is the all-zero vector of length $n$ (it will always
be clear from the context what is the length of the vector
$\zero$). The {\it support} of $x\in \F^n$ is the set of nonzero
coordinate positions of $x$ and is denoted by $\supp(x)$.

Any nonempty subset $\C$ of $\F^n$ is a {\it binary code} and any
vector subspace of $\F^n$ is a {\it binary linear code}. The
elements of $\C$ are called {\it codewords}.
The {\it minimum distance} of $\C$, denoted by $d_\C$, is the smallest Hamming distance
between any pair of different codewords.
Let ${\cal S}_n$ be the symmetric group of permutations of length
$n$. Assume that a permutation $\pi \in  {\cal S}_n$ acts on a
vector $x=(x_1,\ldots,x_n)$ as
$\pi(x)=(x_{\pi^{-1}(1)},\ldots,x_{\pi^{-1}(n)})$. Two binary
codes $\C_1$ and $\C_2$ of length $n$ are said to be {\em
isomorphic} if there exists a coordinate  permutation $\pi \in {\cal
S}_n$ such that $\C_2=\{ \pi(x) :  x\in \C_1 \}$.  They are said
to be {\em equivalent} if there exists a vector $y\in \F^n$ and a
coordinate permutation $\pi \in {\cal S}_n$ such that $\C_2=\{
y+\pi(x) : x\in \C_1 \}$. Although the two definitions above stand
for two different concepts, it follows that two binary linear
codes are equivalent if and only if they are isomorphic \cite{Mac}.

A binary code $\C$ of length $n$ is a {\it perfect 1-error
correcting code} (briefly, {\it perfect code}) if every $x \in
\F^n$ is within distance 1 from exactly one codeword of $\C$. The
perfect codes have length $n=2^m-1$, $2^{n-m}$ codewords and
minimum distance 3. For any integer $m\geq 2$, there exists a
unique, up to equivalence, perfect
linear code of length $n=2^m-1$, called the {\it Hamming code} and
denoted by ${\cH}^n$  \cite{Mac}. Let $H_m$ be a parity check
matrix of the Hamming code ${\cH}^n$ of length $n=2^m-1$. The
columns in $H_m$ are all the nonzero vectors in $\F^m$. We can
associate to each one of them the elements in the set $N=\{1,2,\ldots,n\}$ as well as  the elements $\{\alpha^0, \alpha^1, \ldots,\alpha^{n-2}\}$, where $\alpha$ is a primitive element of the finite field $GF(2^m)$.

Let $F:\F^m \longrightarrow \F^m$ be a function
such that $F(\zero) = \zero$. The function $F$ is called \APN
({\it almost perfect nonlinear}) if all equations
\begin{equation}\label{equationAPN}
F(x) + F(x + b) = a;\,\, a, b \in \F^m; \,\, b\not=\zero,
\end{equation}
have no more than two solutions in $\F^m$. In this paper, we
consider \APN permutations, that is, when the \APN function $F:\F^m
\longrightarrow \F^m$ is bijective, so it corresponds to a
permutation $\pi_F \in {\cal S}_n$, where $n=2^m-1$.
Let $H_F$ be the matrix
\begin{equation}\label{hf}
H_F=\left(\begin{array}{c}H_m\\H_m^{(F)}
\end{array}\right) =\left( \begin{array}{ccc}
\cdots & x &\cdots\\
\cdots & F(x) &\cdots
\end{array}\right),\end{equation}
where $x\in \F^m$, $x\neq \zero$, and let $\C_F$ be the linear code
admitting $H_F$ as a parity check matrix.  Note that $\C_F$ is a
subcode of the Hamming code $\cH^n$. It is known that two functions $F, G :\F^m
\longrightarrow \F^m$, with  $rank(H_F)=rank(H_G)=2m$, are {\it CCZ-equivalent} if and only if
the extended codes $\C^*_F$ and $\C^*_G$ are equivalent. This equivalence relation has been used to
classify \APN functions, since if $F$ is an \APN function and $G$
is \textrm{CCZ}-equivalent to $F$, then $G$ is also an \APN function.

In the last years, many new \APN functions have been constructed
\cite{Const0a, Const1a, Const1b, Const3}. However, it is not
always easy to prove that they are not \textrm{CCZ}-equivalent to any of
the known ones. In order to help to distinguish them,
up to \textrm{CCZ}-equivalence, some invariants have been
defined \cite{Const0b, Const3}.

\medskip
A {\it Steiner triple system} of order $n$ (briefly $STS(n)$) is a
family of 3-element blocks (also called triples) of the set $N$
such that each unordered pair of elements of $N$ appears in
exactly  one block. A $STS(n)$ exists if and only if $n\equiv 1$
or $3\pmod 6$. It is well known that the supports of the codewords of weight 3 in any perfect code
containing the all-zero vector define a Steiner triple system.
For a Hamming code ${\cH}^n$, the
corresponding Steiner triple system is called {\it Hamming Steiner
triple system} and denoted by $STS({\cH}^n)$. Two designs are
called {\it isomorphic} if there is a permutation on the set of
points such that blocks of one design are mapped to blocks of the
other design.

The relation between combinatorial designs and graph embeddings
comes from the fact that when a graph is embedded in a surface,
the faces that results can be regarded as the blocks of a design
\cite{GrannelSurvey}. In the current paper, we consider the case
of a complete graph with $n$ vertices,  embedded into a closed
surface in which all the faces are triangles.
It is  known  \cite{Ringel} that  this complete graph
triangulates some  orientable surface if and only if
 $n\equiv 0, 3, 4$ or $7 \pmod {12}$, and triangulates some
nonorientable surface if and only if $n\equiv 0$ or $1 \pmod 3$
for $n > 7$.

A triangulation is {\it face 2-colorable}  if the triangular faces
of an embedding into a surface  can be properly 2-colored (for
example, in black and white colors),  that is,
in such a way that no two faces with a common edge have the same color.
The case of 2-colorability is of special interest because all the
triangles of the same color on the surface induce a $STS(n)$.
Therefore, we have
two $STS(n)$ (black and white) {\it biembedded} in the surface.
Such a pair of Steiner triple systems of order $n$ is called a
{\it biembedding}. If these two $STS(n)$ are isomorphic, then it
is called a {\it self-embedding}, and the corresponding
permutation is called a {\it self-embedding permutation}.

Two biembeddings are said to be {\it isomorphic} if there exists a permutation on the $n$ vertices
(of the complete graph) such that it maps edges and triangles of one biembedding
to edges and triangles of the other one \cite{Grannell2011, Grannell2002}.
Equivalently, two biembeddings are isomorphic if and only if there exists a permutation on the $n$ vertices
such that it either preserves the color of the triangles or reverse the color.
In the case when the colors of the triangles  are preserved, the isomorphism is said to be {\it color-preserving}.

For an embedding to be face 2-colorable,
$n$ must be odd because  the vertex degrees should be even.
Therefore, for an orientable case, we have that $n\equiv  3$ or $7
\pmod {12}$, and it is known \cite{Ringel,ZT} that if  a
biembedding of a surface exists, the surface should be a sphere
$S_g$ with $g=(n-4)(n-3)/12$ handles. On the other hand, for a
nonorientable case, we have that  $n\equiv 1$ or $3 \pmod 6$ for
$n > 7,$ and therefore  a biembedding of a sphere $N_\gamma$ with
$\gamma=(n-4)(n-3)/6$ crosscaps should exist
 \cite{Ringel,ZT}.

The previous ideas about biembeddings in a closed surface (the
sphere with $g$ handles or  with $\gamma$ crosscaps) can be
extended to pseudosurfaces,
see for example \cite{Kuhnel}. A {\it pseudosurface} is the
topological surface (allowing, in general, repeated triangles)
 which results when finitely many identifications of
finitely many points each, are made  on a given surface. The
points resulting after these identifications are called {\it pinch points}.
All necessary definitions and notions concerning embeddings in
closed surfaces can be found in \cite{Ringel, GrannelSurvey} and
concerning embeddings in pseudosurfaces with pinch points in \cite{Kuhnel,
GrannelSurvey}.
Throughout of what follows, when we refer to self-embeddings,
we always mean self-embeddings in a pseudosurface in general (either a closed surface or pseudosurface with pinch points),
and each time we emphasize if we just deal with a closed surface, that is, a pseudosurface without pinch points.

Despite the existence of many results devoted to
embeddings of a complete graph in a closed surface or
pseudosurface with pinch points, there still remain many unsolved
problems, see the surveys \cite{GrannelSurvey, Kuhnel}.
 For example, it is interesting to find self-embeddings in a closed surface for
the Hamming Steiner triple system $STS({\cH}^n)$ of  order
$n=2^m-1,$ $m>4$. For $n=7$, it is well known that, up to
isomorphism, the $STS({\cH}^7)$ has  only one  self-embedding, which
is a torus and, therefore, is orientable  \cite{Ringel}. For $n=15$, there are
four nonisomorphic self-embeddings of $STS(\cH^{15})$,  three of
them are nonorientable and one is orientable \cite{Grannel}.  On
the other hand, in general, it is easy to obtain self-embeddings  in a
pseudosurface just taking any two isomorphic $STS(\cH^{n})$, or in
general any two isomorphic $STS(n)$, on the same set $N$.

\medskip
In this paper, we only consider self-embeddings, in closed
surfaces and pseudosurfaces with pinch points, obtained from
the Hamming Steiner triple systems $STS({\cH}^n)$ of order $n=2^m-1$, $m>4$, via  monomial power permutations.
We restrict ourselves to
these permutations in order to develop techniques to
find new self-embeddings in  closed surfaces for these $STS({\cH}^n)$ and investigate the connection between
pseudosurfaces and \APN functions which are also monomial power
permutations.

The paper is organized as follows. In Section \ref{sec:1}, we
defined some notions of coding theory (specifically, Hamming
codes and \APN
functions), design theory (specifically, Steiner triple systems), and
graph embeddings into a surface or pseudosurface (specifically,
self-embeddings for Hamming Steiner triple systems). In Section \ref{sec:2},
we present new self-embeddings in closed
surfaces for the Hamming Steiner triple systems $STS({\cH}^n)$,
where $n=2^m-1$ and $m\in \{5,7,11,13,17,19\}$. Actually, we give
all possible self-embeddings in a closed surface constructed from a  $STS({\cH}^n)$
and considering only monomial power permutations, for all $m\leq
22$. Up to isomorphism, there are exactly 1, 1, 4, 14, 12, 65 and 88 such
self-embeddings for $m=3,5,7,11,13,17,19$, respectively.
Note that for any non prime $m$,
there are no such self-embeddings. We also point out
which of all these self-embedding permutations are \APN
permutations. In Section \ref{sec:3}, we focus on showing that the
rotation line spectrum for self-embeddings of  Hamming Steiner triple
systems in pseudosurfaces with pinch points can be used as an
invariant to classify \APN permutations. Actually,
this invariant gives a complete classification of all \APN
monomial power permutations for all $m\leq 17$, up to \textrm{CCZ}-equivalence.
Moreover, it could
be used to classify any \APN permutation, or in general, any
permutation, not necessarily \textrm{APN}.
Finally, in Section \ref{sec:5}, we
present some conclusions and further research.

\section{Self-embeddings of $STS({\cH}^n)$ in  closed surfaces}
\label{sec:2}

In this section, we  construct new self-embeddings in  closed
surfaces for Hamming Steiner triple systems $STS({\cH}^n)$, where $n=2^m-1$ and $m\in \{5,7,11,$ $13,$ $17,19\}$.
Moreover, up to isomorphism, we give all possible such self-embeddings for the $STS({\cH}^n)$, with $m\le 22$, constructed
from monomial power permutations, together with their classification.

Without loss of generality, we can consider a cyclic $STS({\cal H}^n)$
corresponding to a cyclic Hamming code ${\cH}^n$ of length $n$. It
is known that the cyclic $STS({\cal H}^n)$ is unique for every
$n=2^m-1$, up to isomorphism. A design defined on the set $N$ is
called {\it cyclic} if there is a permutation on the set $N$
consisting of a single cycle of length $n$ such that blocks are
mapped to blocks.

It is easy to see that there is only one self-embedding in a closed surface for the cyclic
$STS({\cH}^7)$ via the permutation corresponding to
the monomial power function $F(x)=x^3$ over $\F^3$. For
$n=15$, none of the four nonisomorphic self-embeddings of
$STS(\cH^{15})$ classified in \cite{Grannel} are cyclic,
so there are no self-embeddings in a closed surface for the cyclic
$STS(\cH^{15})$ given by monomial power permutations. For $n=31$,
Bennett at al. proved that there is
not any cyclic orientable self-embedding in a closed surface for the $STS(\cH^{31})$ \cite{Bennett}.
It is still an open question to determine whether there exist noncyclic orientable
self-embeddings in a closed surface  for the $STS(\cH^{31})$ or not.
In this section, we present new self-embeddings for the cyclic
$STS({\cH}^n)$ with $n=2^{m}-1$, which are cyclic for all $m$
and nonorientable at least for all $m\leq 19$.

Note that there are $STS(n)$ which are not isomorphic to the
$STS({\cal H}^n)$ but also have permutations without fixed points in their automorphism group.
For example, the $STS(n)$ given by the well known Bose construction \cite{MHall}
has an automorphism group containing a  permutation with three
short cycles of length $n/3$. An interesting fact is  that Bose $STS(15)$
can not be included in any perfect code of length $15$, see
\cite{OP}. There are several constructions of self-embeddings for
the $STS(n)$ obtained from the Bose construction, orientable and
nonorientable \cite{GrannelSurvey,Sol}.

In order to construct these mentioned new self-embeddings in a closed surface for the cyclic $STS({\cal H}^n)$,
we only consider permutations $\pi_F \in {\cal S}_n$, where $n=2^m-1$,
given by monomial power functions $F(x)=x^t$ over $\F^m$, so such that
$gcd(t,n)=1$. Therefore, since the $STS({\cal H}^n)$ is cyclic,
these constructed self-embeddings are also cyclic,
according to the next proposition.

\begin{proposition}\label{def:Cyclicity}
Let $STS({\cal H}^n)$ be cyclic and $F:\F^m \longrightarrow \F^m$
be any monomial power permutation.
Then,  $F(STS({\cal H}^n))$ is also cyclic.
\end{proposition}

\begin{proof}
Let $\alpha$ be a primitive element of the finite field $GF(2^m)$.
For any triple $(\alpha^i, \alpha^j, \alpha^k)$ from the cyclic
$STS({\cal H}^n)$ corresponding to the Hamming code with
parity check matrix $H_m=(1 \, \alpha \, \alpha^2 \, \ldots \, \alpha^{n-1})$,
we have $F(\alpha^i, \alpha^j, \alpha^k)=(\alpha^{it},
\alpha^{jt}, \alpha^{kt})\in F(STS({\cal H}^n)),$  where $F(x)=x^t$. Since $gcd(t,n)=1$,
all the numbers $it,(i+1)t, \ldots, (i+n-1)t$ are different and
run through the set $N=\{1,2,\ldots,n\}$. Therefore, applying the monomial power permutation $F$  to all
different triples from the cyclic $STS({\cal H}^n)$, we obtain the
different triples $(\alpha^{(i+l)t}, \alpha^{(j+l)t},
\alpha^{(k+l)t})\in F(STS({\cal H}^n))$, where $l\in
\{0,1,\ldots,n-1\}$. Hence, $F(STS({\cal H}^n))$ is cyclic.
\end{proof}

\medskip
In general, a biembedding in a pseudosurface has a pinch point if and only if
there is a point $i\in N$ such that the cyclically ordered points
of all triples containing $i$ in both black and white $STS(n)$,
with the ordering determined by the biembedding, can be divided
into more than one cycle. Each one of these cycles is called {\it
rotation line} at point $i\in N$. Note that a biembedding in a closed
surface has no pinch points, so the rotation line at each point
contains a single cycle of length $n-1$. We collect all
rotation lines at point $i\in N$ taking them in any order. The
number of rotation lines at point $i\in N$ will be denoted by
$rl(i)$. A biembedding in a closed surface can be considered as a
biembedding in a pseudosurface such that $rl(i)=1$ for any $i\in
N.$ The set of rotation lines at all the points of $N$ is called
the {\it rotation scheme} for the biembedding.

\medskip
The next proposition gives us an alternative definition for a
self-embedding permutation in a closed surface for a $STS({\cal H}^n)$.
Given a $STS({\cal H}^n)$, where $n=2^m-1$, for all $a,b \in \F^m
\backslash \{\zero \}$ with $a \not = b$, we have $a + b = c$ if
$(a,b,c)$ is a triple in $STS({\cal H}^n)$. Note that, from now on,
we will use  indistinctly the vectors in $\F^m \backslash \{\zero \}$
as elements (points) of the $STS({\cal H}^n)$ and vice versa.

\begin{proposition}\label{def:SE}
Let $F$ be any bijective function over $\F^m$ such
that $F(\zero) = \zero$. The permutation $F$ is a self-embedding permutation
in a closed surface for the $STS({\cal H}^n)$ if
and only if, for any $a \in \F^m \backslash \{ \zero \}$, the
elements in the sequence $a_1, a_2, \ldots, a_{2^{m-1}-1}$  are
different elements in $\F^m$, where $a_1$ is any element in $\F^m
\backslash \{ \zero \}$ such that  $a_1 \not = a$ and
$a_{i+1}=F(F^{-1}(a)+F^{-1}(a+a_i))$ for all $i\in\{1,\ldots,
2^{m-1}-2\}$.
\end{proposition}

\begin{proof}
Given the self-embedding permutation $F$ in a closed surface, the rotation line at any
element $a\in \F^m \backslash \{ \zero \}$ is a sequence of
$2^m-2$ different elements:
$$ [a_1,a+a_1;a_2,a+a_2; \ldots; a_{2^{m-1}-1}, a + a_{2^{m-1}-1}], $$
where $a_1$ is any element in $\F^m \backslash \{ \zero \}$ such
that $a_1 \not = a$, and $(a,a_i, a+a_i)$ is a triple in
$STS({\cal H}^n)$ for all $i\in\{1,\ldots, 2^{m-1}-1\}$. Note that
$F(STS({\cal H}^n))$ is a Steiner triple system isomorphic to
$STS({\cal H}^n)$. Moreover, the blocks in $F(STS({\cal H}^n))$
can be seen as the triples $(a, b, F(F^{-1}(a)+F^{-1}(b)))$, where
$+$ stands for the operation defined above for the $STS({\cal
H}^n)$. Therefore, for all $i\in\{1,\ldots, 2^{m-1}-2\}$, taking
$b=a+a_i$, we have that $a_{i+1}=F(F^{-1}(a)+F^{-1}(a+a_i))$.
\end{proof}

We have used Proposition \ref{def:SE} to find new self-embedding
permutations in closed surfaces for the cyclic  $STS({\cal H}^n)$, where $n=2^m-1$ with $m\leq 22$.
Note that,  considering permutations $\pi_F \in {\cal S}_n$
given by a monomial power function $F(x)=x^t$ over $\F^m$ such that
$gcd(t,n)=1$, it is only necessary to check the
condition for just one element $a \in \F^m \backslash \{\zero\}$.

\medskip
For any self-embedding of  $STS({\cal H}^{n})$ given by a permutation
$F:\F^m \longrightarrow \F^m$ with $F(\zero)=\zero$,
and any element $a\in \F^m \backslash \{ \zero \}$, we can
construct the sequence
  $a_1,a_2,\ldots, a_{r_1}$ beginning with any
element $a_1 \in \F^m \backslash \{ \zero \}$ such that $a_1 \not
= a$ and $a_{r_1+1}=a_1$, where
\begin{equation}\label{def:a_i+1}a_{i+1}=F(F^{-1}(a)+F^{-1}(a+a_i)).\end{equation}
Let $b_i$ be the element in
$\F^m \backslash \{ \zero \}$ such that $(a,a_i,b_i)$ is a triple
for all $i\in \{1,\ldots,r_1\}$. Then, the sequence
$R_1=[a_1,b_1;a_2,b_2; \ldots; a_{r_1},b_{r_1}]$ define a
rotation line at point $a$. If the rotation line $R_1$ does not
cover all elements in $\F^m \backslash \{ \zero \}$, we take an
element out of the rotation line and construct another rotation
line $R_2$ beginning with this point, and so on. Finally, we obtain a
partition of all elements in $\F^m \backslash \{ \zero, a \}$ in
different rotation lines $R_1, R_2, \ldots, R_s$, where $s=rl(a)$.
We simplify this information considering only the number of
rotation lines and the cardinal of each one of them. The
\textit{rotation line spectrum} at point $a$ will be the array
\begin{equation}
(s; rl(a)_1, rl(a)_2,\ldots , rl(a)_s),
\end{equation}
where $s=rl(a)$ is the number of rotation lines at point $a$, and
$rl(a)_i=|R_i|$  for all $i\in \{1,\ldots,s\}$.
Note that the rotation line spectrum of a self-embedding permutation in a closed surface for the
$STS({\cal H}^n)$  is $(1;n-1)$, where $n=2^m-1$.

\begin{ex} \label{ExamplesRotationLines}
For $m=5$, consider the cyclic $STS({\cal H}^{31})$ corresponding to the
cyclic Hamming code with parity check matrix $H_5=(1 \ \alpha \ \alpha^2 \ \ldots \ \alpha^{30} )$,
where $\alpha$ is a primitive element of the finite field $GF(32)=\F[x]/(x^5+x^2+1)$.

The permutation $\pi_F=(2,26,6)(3,20,11)(4,14,16)(5,8,21)(7,27,31)(9,$ $15,$
$10)(12,28,25)(13,22,30)(17,29,19)(18,23,24)
\in {\cal S}_{31}$, which corresponds to the bijective function $F(x)=x^5$ over $\F^5$,
is a self-embedding permutation in a closed surface for the $STS({\cal H}^{31})$.
The rotation line at point 1 is given by the sequence
\begin{equation}\label{ex:rotationlineS1}
\begin{split}
R_1=&[2,19;31,18;22,26;17,10;16,25;27,29;9,21;\\
&24,13;14,15;5,11;28,7;23,8;3,6;30,4;12,20],
\end{split}
\end{equation}
 so the rotation line spectrum is $(1;30)$.

On the other hand, the permutation  corresponding to the function $F(x)=x^3$ over $\F^5$
is a self-embedding permutation in a pseudosurface with pinch points
for the $STS({\cal H}^{31})$. Note that in this case there are two rotation lines
at point 1 given by the sequences
\begin{equation}
\begin{split}
R_1=&[2,19;5,11;17,10;3,6;9,21], \\
R_2=&[ 27, 29; 24, 13; 12, 20; 31, 18; 14, 15;
28,7; 22, 26; 16, 25; 23, 8; 30, 4],
\end{split}
\end{equation}
 so the rotation line spectrum is $(2; 10,20)$.
\end{ex}

For any self-embedding of  $STS({\cal H}^{n})$ given by a permutation
$F:\F^m \longrightarrow \F^m$ with $F(\zero)=\zero$,
we can calculate how many different values there are in the set
$\{x+F^{-1}(a+F(x)) \,:\, x\in \F^m\}$
for any $a\in \F^m\backslash \{\zero\}$. Let
\begin{equation}
v_F(a)=|\{x+F^{-1}(a+F(x)) \,:\, x\in \F^m\}|.
\end{equation}  We can also calculate the multiset $\tilde{V}_F(a)= \{ z_i : i \in \{1,\ldots, 2^{m-1}-1 \} \}$,
where $\{ (a, a_i, a+a_i) :  i \in \{1,\ldots, 2^{m-1}-1 \} \}$ is the set of all triples in $STS({\cal H}^{n})$
containing the point $a$ and $(a_i, a+a_i,z_i)$ are triples in $F(STS({\cal H}^{n}))$ for all $i \in \{1,\ldots, 2^{m-1}-1 \}$.
Let $V_F(a)$ be the set associated to $\tilde{V}_F(a)$, and let $V^*_F(a)$ be the multiset containing
the multiplicities of the different elements in $\tilde{V}_F(a)$. We denote by $x^\wedge s$ the elements in $V^*_F(a)$, understanding that we have $s$ different elements in  $\tilde{V}_F(a)$ appearing $x$ times.
In the following lemma, we give the connection
between $v_F(a)$ and $V_F(a)$, and  then we show these definitions by Example \ref{v_V}.

\begin{lemma} \label{lem:VF}
Let $F$ be any bijective function over $\F^m$ such that $F(\zero)=\zero$, and let $S=STS({\cal H}^{n})$.
If $S \cup F(S)$ is a self-embedding, then $v_F(a)=1+ |V_F(a)|$.
\end{lemma}

\begin{proof}
We have that $\{x+F^{-1}(a+F(x)) \,:\, x\in \F^m\}= \{F^{-1}(y)+F^{-1}(a+F(F^{-1}(y))) \,:\, y\in \F^m\}=
\{F^{-1}(y)+F^{-1}(a+y) \,:\, y\in \F^m\}$, where $y=F(x)$. Therefore, $$v_F(a)=|\{ F^{-1}(a) \} \cup  \{  F^{-1}(z_i) : i \in \{1,\ldots, 2^{m-1}-1\}|,$$ where $z_i=F( F^{-1}(a_i)+F^{-1}(a+a_i) )$ and
$a_1,a+a_1;a_2,a+a_2; \ldots; a_{2^{m-1}-1}, a + a_{2^{m-1}-1}$
is the sequence containing all the rotation lines at point $a$.
Note that regardless having a self-embedding in a closed surface or pseudosurface with pinch points,
we just consider the triples $(a_i, a+a_i,z_i)$, for all $i \in \{1,\ldots, 2^{m-1}-1\}$,
which are blocks in  $F(S)$.
Moreover, since $a\not = z_i$ for all $i \in \{1,\ldots, 2^{m-1}-1\}$ and $F$ is bijective,
$F^{-1}(a) \not =F^{-1}(z_i)$ and $v_F(a)=1+|\{  F^{-1}(z_i) : i \in \{1,\ldots, 2^{m-1}-1\}|=
1+|\{  z_i : i \in \{1,\ldots, 2^{m-1}-1\}|=1+|V_F(a)|.$
\end{proof}

\begin{ex}\label{v_V}
Consider the cyclic $STS({\cal H}^{127})$ corresponding to the
cyclic Hamming code with parity check matrix $H_7=(1 \ \alpha \ \alpha^2 \ \ldots \ \alpha^{126} )$,
where $\alpha$ is a primitive element of the finite field $GF(128)=\F[x]/(x^7+x+1)$.

For the self-embedding in a closed surface, given by the permutation
$F(x)=x^{7}$ over $\F^7$,
we have that

\smallskip
\noindent
{\small
$\tilde{V}_F(1)=\{ 109, 43, 17, 28, 56, 40, 103, 82, 64, 78, 38, 3, 52, 119, 117, 109, 27, 120,
90, 85,\\ 33, 55, 111, 79, 78, 36, 127, 28, 75, 5, 103, 110, 106, 90, 53, 112, 52,
42, 65, 109, 94, 30,\\ 28, 71, 126, 55, 22, 9, 78, 92, 84, 52, 105, 96, 103, 83,
2, 90, 60, 59, 55, 14, 124 \},$
}

\smallskip
\noindent since the rotation line at point 1
is $R_1=[2,8;91,10;  74, 79; \ldots;36,110]$ and $(2,8,109),$ $(91,10,43),$ $(74,79,17), \ldots, (36,110,124)$ are triples in $F(STS({\cal H}^{127}))$. Therefore, by Lemma \ref{lem:VF} (see also Table \ref{taula_multiplicities1}), $$v_F(1)= 1 +  |V_F(1) |=50 \quad \textrm{and} \quad  V^*_F(1)=\{ 1^\wedge 42, 3^\wedge{7} \}.$$
\end{ex}

\begin{lemma}\label{invers}
Let $F$ be any bijective function over $\F^m$ such that $F(\zero)=\zero$, and let $S=STS({\cal H}^{n})$.
If $S \cup F(S)$ is a self-embedding, then $S \cup F(S)$ and $S \cup F^{-1}(S)$  are isomorphic.
\end{lemma}
\begin{proof}
It is easy to check that the permutation $F$ transforms the triples from $S$ into the triples in $F(S)$ and the triples from $F^{-1}(S)$ into the triples in $S$.
\end{proof}

\begin{theorem} \label{theo:invariant}
Let $F_1, F_2$ be two bijective
functions over $\F^m$ such that $F_1(\zero)=F_2(\zero)=\zero$, and let $S=STS({\cal H}^{n})$. If  $S \cup F_1(S)$ and $S \cup F_2(S)$
 are isomorphic self-embeddings, then
$$\{ v_{F_1}(a) : a \in \F^m \backslash \{\zero\} \}=\{ v_{F_2}(a) : a \in \F^m \backslash \{\zero\} \}, \ \textrm{and}$$
$$\{ V^*_{F_1}(a) : a \in \F^m \backslash \{\zero\} \}=\{ V^*_{F_2}(a) : a \in \F^m \backslash \{\zero\} \}.$$
\end{theorem}

\begin{proof}
By Lemma~\ref{invers}, it is enough to assume that the isomorphism is given by a function $F$ transforming triples into triples such that $F(S)=S$ and $F(F_1(S))=F_2(S)$. Looking at the points as vectors in $\F^m \backslash \{\zero\}$, we can consider  the function $F$ as a linear transformation on $\F^m$.

The elements in the set $V_{F_1}(a)$ are $F^{-1}_1(z_i)$, for $i \in \{1,\ldots, 2^{m-1}-1\}$, where $z_i=F_1( F^{-1}_1(a_i)+F^{-1}_1(a+a_i) )$ and $(a_i, a+a_i,z_i)$ are the triples in  $F_1(S)$. For any triple $(a_i, a+a_i,z_i)$ in $F_1(S)$, we have that $(F_2(F^{-1}_1(a_i)), F_2(F^{-1}_1(a+a_i)),F_2(F^{-1}_1(z_i)))$ is a triple in $F_2(S)$ and so, as $F_2\circ F^{-1}_1=F$, we see that $(F(a_i), F(a+a_i),F(z_i)) = (F(a_i), F(a)+F(a_i),F(z_i))$ are the corresponding triples in $F_2(S)$. Since $F$ is bijective,  we conclude that $\tilde{V}_{F_1}(a)=\tilde{V}_{F_2}(F(a))$, $V^*_{F_1}(a)=V^*_{F_2}(F(a))$  and using Lemma~\ref{lem:VF} we obtain $v_{F_1}(a)=v_{F_2}(F(a))$ for all $a \in \F^m \backslash \{\zero\}$. Therefore, the result follows.
\end{proof}

\begin{proposition} \label{prop:monomialPower_Permutation}
Let $F: \F^m \longrightarrow \F^m$  be any monomial power permutation, and let $S=STS({\cal H}^{n})$.
If  $S \cup F(S)$ is a self-embedding,
then the parameters $v_F(a)$ and $V^*_F(a)$ do not depend on the choice of $a\in \F^m
\backslash \{\zero\}$, that is, $v_{F}(a)=v_F(1)$ and
$V^*_{F}(a)= V^*_F(1)$ for all $a \in \F^m \backslash \{\zero\}.$
\end{proposition}

\begin{proof}
If $F$ is a monomial power permutation, then at each point $a\in \F^m \backslash \{\zero\}$
the rotation line is the same, up to a permutation and so $\tilde{V}_{F}(a)$ and $\tilde{V}_{F}(1)$ also coincide, up to a permutation.
Finally, by the definitions of $v_F(a)$, $V^*_{F}(a)$ and Lemma \ref{lem:VF}, the result follows.
\end{proof}

\begin{ex}
For the self-embedding permutation in a closed surface, given by the permutation
$F(x)=x^5$ over $\F^5$ defined in Example \ref{ExamplesRotationLines},
we have that $$\tilde{V}_F(1)=V_F(1)=\{ 23, 30, 5, 12, 31, 3, 22, 16, 2, 27,24, 17, 14, 28, 9 \}$$
since the rotation line at point 1
is $R_1$ given in (\ref{ex:rotationlineS1}), and $(2,19,23),$ $(31,18,30),$ $(22,26,5), (17,10,12), (16,25,31),(27,29,3),(9,21,22)$,
$(24,13,16),$ $(14,15,\allowbreak 2)$, $(5,11,27),(28,7,24),(23,8,17),(3,6,14),(30,4,28),(12,20,9)$ are triples in $F(STS({\cal H}^{31}))$. Therefore,
$$v_F(1)= 1 +  |V_F(1) |=16 \quad \textrm{and} \quad  V^*_F(1)=\{ 1^\wedge 15 \}. $$
Finally, by Proposition \ref{prop:monomialPower_Permutation}, $v_F(a)=v_F(1)=16$ and $V^*_F(a)=V^*_F(1)=\{ 1^\wedge 15 \}$ for all $a\in \F^m \backslash \{\zero\}$.

For the self-embedding in a pseudosurface with pinch points, given by the permutation $F(x)=x^3$ over $\F^5$ defined in Example \ref{ExamplesRotationLines}, we also have that $v_F(a)=v_F(1)=16$ and $V^*_F(a)=V^*_F(1)=\{ 1^\wedge 15 \}$ for all $a\in \F^m \backslash \{\zero\}$.
\end{ex}

Note that $v_F(a)$ is maximum when all the elements in $\tilde{V}_F(a)$ are different, that is, when  $v_F(a)=2^{m-1}$.
For both permutations in the previous example, $v_F(a)$ is maximum for all $a\in \F^m \backslash \{\zero\}$. Therefore, by Proposition \ref{SE_APN} (see Section \ref{sec:3}, where we investigate the connection between \APN functions and self-embeddings), they are \APN permutations.

\medskip
By Theorem \ref{theo:invariant}, we have that the sets
$\{ v_{F}(a) : a \in \F^m \backslash \{\zero\} \}$ and $\{  V^*_F(a) : a \in \F^m \backslash \{\zero\} \}$
can be used as invariants to distinguish
nonisomorphic self-embedding permutations $F$.
By Proposition \ref{prop:monomialPower_Permutation}, note that considering  monomial
power permutations, it is only necessary to compute $v_F(a)$ and $V^*_F(a)$ for  one
element $a\in \F^m \backslash \{\zero\}$, for example $a=1$.
Let $v_F=v_F(1)$ and $V^*_F=V^*_F(1)$.
We further use these invariants to classify the found
self-embedding permutations in  closed surfaces.

\medskip
Let $C_i$ be the (binary) cyclotomic coset containing $i$, that is, the set
of integers $C_i=\{i, 2i, 4i, \ldots, 2^{m_i-1}i \}$, where $m_i$
is the smallest positive integer such that $2^{m_i} \cdot i \equiv
i \ (\textrm{mod} \ 2^m-1)$ \cite{Mac}.
The cyclotomic cosets give a partition of the integers modulo $2^m-1$
into disjoint subsets. Let $C_i^*$ be the union of the cyclotomic
coset containing $i$ and the cyclotomic coset containing the
multiplicative inverse of $i$ modulo $2^m-1$. Note that in some
cases the set $C_i^*$ coincides with $C_i$,
for example, $C_1^*=C_1=\{1, 2, 4, \ldots, 2^{m-1} \} $.

The following result demonstrates that if $t_1$ and $t_2$ are in
the same set $C_i^*$, the self-embedding permutations corresponding to
$F_1(x)=x^{t_1}$ and $F_2(x)=x^{t_2}$ are isomorphic and have the same parameters $V^*_{F_1}(a)=V^*_{F_2}(a)=V^*$ and
$v_{F_1}(a)=v_{F_2}(a)=v$, which are fixed for all $a\in \F^m \backslash \{\zero\}$.

\begin{proposition} \label{MonomialPowerPerm}
Let $F_1, F_2: \F^m \longrightarrow \F^m $ be two monomial power
permutations $F_1(x)=x^{t_1}$ and $F_2(x)=x^{t_2}$ such that $t_1,
t_2 \in C_i^*$, and let $S=STS({\cal H}^{n})$. If $S \cup F_1(S)$ and $S \cup F_2(S)$
are two self-embeddings, then they are isomorphic, $V^*_{F_1}(a)=V^*_{F_2}(a)=V^*$
and $v_{F_1}(a)=v_{F_2}(a)=v$ for all $a\in \F^m
\backslash \{\zero\}$.
\end{proposition}

\begin{proof}
The Frobenious automorphisms $F(x)=x^{2^s}$ over $\F^m$, for
$s\in \{1,2,\ldots,m-1\}$, are well-known examples of permutations $\pi_F\in {\cal S}_n$ transforming
$S=STS({\cal H}^n)$ into itself. Indeed, the triples $(a,b,c)$ of
$S$ are such that $a+b=c$, where $a,b,c \in \F^m
\backslash \{\zero \}$, so $(a+b)^{2^s}=a^{2^s}+b^{2^s}=c^{2^s}$ giving that
$(F(a),F(b),F(c))$ are also triples in $S$.

Let $t_1 \in C_i$ and $t_2\in C_i$. Using a Frobenious automorphism $F(x)=x^{2^s}$ for
some $s\in \{1,2,\ldots,m-1\}$, we
have that $F_2=F\circ F_1$. Therefore, the two self-embeddings
$S \cup F_1(S)$ and $S \cup F_2(S)$ are isomorphic.

Let  $t_1\in C_i$ and $t_2\in C_j$, where $C_i$ and $C_j$ are the
inverse cyclotomic cosets such that  $C_i^*=C_i\cup C_j$.
Up to a Frobenious automorphism $F(x)=x^{2^s}$ for some $s\in
\{1,2,\ldots,m-1\}$, we can assume that $t_1$ is the multiplicative inverse
of $t_2$ modulo $2^m-1$. Hence,
$(x^{t_1})^{t_2}=x$, which means that $F_2 (F_1(x))=x$, and the
corresponding permutations $\pi_{F_1}$ and $\pi_{F_2}$ satisfy $\pi_{F_2}= \pi_{F_1}^{-1}$.
Therefore, in general, $F_2=F_1^{-1} \circ F$, and again the two self-embeddings
$S \cup F_1(S)$ and $S \cup F_2(S)$ are isomorphic.

Finally, by Theorem \ref{theo:invariant} and Proposition \ref{prop:monomialPower_Permutation},
we have that $V^*_{F_1}(a)=V^*_{F_2}(a)=V^*$ and
$v_{F_1}(a)=v_{F_2}(a)=v$ for all $a\in \F^m \backslash \{\zero\}$.
\end{proof}

\begin{proposition}\label{prime}
For any non prime $m$, there is not any self-embedding in a closed surface for the $STS({\cal H}^{n})$
 given by a monomial power permutation.
\end{proposition}
\begin{proof} Let us assume there is a self-embedding in a closed surface for the $STS({\cal H}^{n})$
given by a permutation $F(x)=x^t$, that is, such that $gcd(t,n)=1$, where $n=2^m-1$.
Therefore, we can construct the rotation line at point 1 as
$$[a_1,1+ a_1;a_2,1+ a_2;\ldots;a_{2^{m-1}-1},1+ a_{2^{m-1}-1}].$$

Let $m'$ be a divisor of $m$ such that $1 < m'< m$. Then, $n'= 2^{m'}-1$ divides $n=2^{m}-1$. Let $b=n/n'$ and  $\alpha$ be a primitive element in $GF(2^m)$. Hence, $\alpha^{b}$ is a primitive element in a subfield $K\subset GF(2^m)$ of $2^{m'}$ elements. Since $F$ is
a permutation, $gcd(t,n)=1$, and we have that $F(\alpha^b)=\alpha^{tb}$ generates a finite field of $2^{m'}$ elements which coincides with $F(K)\subset GF(2^m)$.

We can consider $a_1=\alpha^{tb}\in F(K)$, where $(1,a_1,1+ a_1)=(1,\alpha^{tb},1+\alpha^{tb})\in F(STS({\cal H}^{n}))$.
If $a_i \in F(K)$, then $a_{i+1}=F(F^{-1}(1)+F^{-1}(1+a_i)) \in F(K)$ by (\ref{def:a_i+1}). Therefore, we have that $a_i \in F(K)$ for all $i=1,2,\ldots,2^{m-1}-1$. By Proposition~\ref{def:SE}, these elements, together with the element $1$  and elements $1+a_i$, for all $i=1,2,\ldots,2^{m-1}-1$, are different elements in $F(K)\backslash \{\zero\}$. Since  $|F(K)|=2^{m'}$ and $m' < m$, this leads to a contradiction.
\end{proof}

\begin{theorem} \label{theo:classification}
For $m\in \{3,5,7,11,13,17,19 \}$, up to isomorphism, there are exactly
1, 1, 4, 14, 12, 65 and 88 self-embedding monomial power permutations in closed surfaces for the $STS({\cal H}^n)$, where $n=2^m-1$,
respectively. Moreover, at least for all $5 \leq m\leq 19$ these self-embeddings are nonorientable.
\end{theorem}

\begin{proof} Using Proposition \ref{def:SE} and the {\sc Magma} software package \cite{M1},
we found all self-embedding permutations for the $STS({\cal H}^n)$ in closed surfaces,
where $n=2^m-1$ and $m\leq 19$, given by monomial power permutations, $F(x)=x^t$ over $\F^m$ such that $gcd(t,n)=1$.
We also computed the parameter $v_F(a)$ for all found self-embedding
permutations $F$ and all $a\in \F^m \backslash \{\zero\}$. By Propositions \ref{prop:monomialPower_Permutation} and \ref{MonomialPowerPerm},
we just had to compute $v_F$ for one representative element
of the set $C^*_t$.  In Appendix (Tables \ref{taulavb1} and \ref{taulavb2}), all these values are listed.
By Proposition \ref{prime}, for any non prime $m$, there is not any self-embedding monomial power permutation in a closed surface.

Since the parameter $v_F$ is an invariant,
by Theorem \ref{theo:invariant}, the self-em\-bed\-dings having different parameters $v_F$ in Tables
\ref{taulavb1} and \ref{taulavb2} are nonisomorphic.
For $m=3$ and $m=5$, this parameter gives a
complete classification, since there is only one class. It is well known that the case $F(x)=x^3$ for $m=3$ is the torus given by
$STS({\cal H}^7)$.
For $m=7$, the self-embedding permutations $F_1(x)=x^7$ and $F_2(x)=x^{21}$ have the same parameter
$v_{F_1}=v_{F_2}=50$. However, using {\sc Magma}, it is easy to check that
the corresponding self-embeddings are not isomorphic, so there are exactly 4 classes of nonisomorphic
such self-embeddings.

For the classes having the same parameter $v_F$, we computed $V^*_F$, which is
also an invariant, by Theorem \ref{theo:invariant}.
For $m=11$, the self-embedding permutations given by $F_1(x)=x^{21}$ and $F_2(x)=x^{687}$ over $\F^{11}$ have the same parameters $v_{F_1}=v_{F_2}=815$, and $V^*_{F_1}=V^*_{F_2}= \{ 1^\wedge{628}, 2^\wedge{165}, 3^\wedge{22} \}$. However, using {\sc Magma}, we checked that the weight distributions of codes $\C_{F_1}$ and $\C_{F_2}$ are different and, by Proposition~\ref{SE-Codes}, we can conclude that the two self-embeddings are nonisomorphic.
On the other hand, for the self-embedding permutations $F_1(x)=x^{73}$ and $F_2(x)=x^{165}$ over $\F^{11}$, which also have the
same parameter $v_{F_1}=v_{F_2}=826$, just using that $V^*_{F_1}\not =V^*_{F_2}$, we have that they are nonisomorphic.
Therefore, there are exactly 14 nonisomorphic such self-embeddings.
For $m=13,17$ and $19$, Tables \ref{taula_multiplicities1} and \ref{taula_multiplicities2} show the parameter $V^*_F$ for
the sets $C^*_t$ having the same parameter $v_F$ in Tables \ref{taulavb1} and \ref{taulavb2}.
Note that all classes can be distinguish, either
using just the invariant $v_F$ or using also the invariant $V^*_F$.

Finally, the computer search using the {\sc Magma} software package showed that the obtained self-embeddings are nonorientable at least for all $m\leq 19$. Therefore, the result follows.
\end{proof}

As far as we know, all found self-embedding permutations in closed surfaces given by Theorem \ref{theo:classification} are
new with the exception of the one given for $m=3$. By Proposition \ref{def:Cyclicity}, these self-embeddings are cyclic for all $m$.

Moreover, note that for every $m\in \{
3,5,7,11,17\}$, there exists a cyclotomic coset $C_i^*$ such that
for all permutations $F(x)=x^t$ with $t\in C_i^*$,
$v_F=2^{m-1}$ is maximum, so the corresponding self-embedding
permutations are also \textrm{APN}. In Table \ref{taulavb1}, we
point out these cases with $({\cdot })^{APN}$. Note that for any non prime $m$
and for $m\in \{13,19\}$ there are no \APN self-embeddings in a closed surface.

\section{Self-embeddings of $STS({\cH}^n)$ and \APN permutations}
\label{sec:3}

This section deals with \APN permutations, which can be seen as self-em\-bed\-dings permutations in a
pseudosurface without triples in common.
Given an \APN function $F$, the corresponding code
$\C_F$ has minimum distance 5. In fact, $F$ is an \APN function if
and only if $\C_F$ has minimum distance 5 \cite{CCZ}. Therefore,
since $\C_F={\cH}^n \cap \pi_F({\cH}^n)$, any \APN permutation $F$
gives two nonintersecting Hamming Steiner triple systems,
$STS({\cH}^n)$ and  $F(STS({\cH}^n))$, which can be seen as a
self-embedding in a closed surface or in a
pseudosurface with pinch points (and without triples in common).

As in the previous section, we  consider the (cyclic) Hamming Steiner triple system $STS({\cH}^n)$
and permutations  $\pi_F \in {\cal S}_n$, where $n=2^m-1$,
given by monomial power functions $F(x)=x^t$ over $\F^m$, so such that $gcd(t,n)=1$.
In this case, we show that the rotation line spectrum of the corresponding
self-embeddings in pseudosurfaces can be used as an invariant to
distinguish between classes of \APN permutations or, in general, to
classify permutations. Moreover, we see that the rotation line spectrum gives a complete
classification of monomial power permutations up to
\textrm{CCZ}-equivalence, at least for all $m\leq 17$, so we can say
that this classification coincides with the one given by the
self-embedding isomorphism. Actually, the invariants
$v_F$ and $V^*_F$ given in Section \ref{sec:2}, can be also used to
distinguish between \textrm{CCZ}-equivalent classes of monomial power
permutations, not necessarily \textrm{APN}.

The next proposition gives a characterization of the \APN permutations using the parameter $v_F(a)$, defined in the previous section for all $a\in \F^m \backslash \{\zero\}$.

\begin{proposition}\label{SE_APN}
Let $F$ be any bijective function over $\F^m$ such
that $F(\zero) = \zero$. The permutation $F$ is \APN if and only
if, for all $a\in \F^m \backslash \{\zero\}$, we have that $v_F(a)=2^{m-1}$.
\end{proposition}
\begin{proof}
Given $ a\in \F^m \backslash \{\zero\}$, assume that $v_F(a) = |\{x+F^{-1}(a+F(x)) \,:\, x\in \F^m\}|=2^{m-1}$. This means that there are
$2^{m-1}$ different values $b\in \F^m \backslash \{\zero\}$, such that the equation
\begin{equation}\label{eq:1}
x+F^{-1}(a+F(x)) =b
\end{equation}
has two solutions, and there is no solution for the other values of
$b$. Since (\ref{eq:1}) is equivalent to (\ref{equationAPN}), $F$
is an \APN permutation.

Vice versa, assume that $F$ is an \APN permutation, and the values
$x+F^{-1}(a+F(x))$ are not all different (up to a multiplicity of
two), for a fixed $a\in \F^m \backslash \{\zero\}$. Hence, there
exists a value $b\neq \zero$ such that (\ref{eq:1})  has more
than two solutions. Again, from (\ref{eq:1}), we obtain that
(\ref{equationAPN}) has more than two solutions, which contradicts
the assumption about $F$ being \textrm{APN}.
\end{proof}

The concept of \textrm{CCZ}-equivalence is not so finer than the concept
of \se equivalence as we show in the next two propositions.

\begin{proposition}\label{SE-Codes} Let $F_1$ and $F_2$ be two bijective functions over $\F^m$ such
that $F_1(\zero)=\zero$ and  $F_2(\zero)=\zero$.
If $F_1$ and $F_2$ are isomorphic self-embedding permutations for the $STS(\cH^n)$, then
the corresponding codes $\C_{F_1}$ and $\C_{F_2}$ are equivalent.
\end{proposition}

\begin{proof}
By Lemma~\ref{invers}, it is enough to assume that the isomorphism is given by a function
$F$ transforming triples into triples such that $F(S)=S$ and $F(F_1(S))=F_2(S)$.
Hence, there exists a linear function $\lambda: \F^m \longrightarrow \F^m$ such that $F\circ F_1=F_2 \circ \lambda$.
We know that $H_m^{(F)}$ (and $\lambda(H_m)$) gives a different matrix than
$H_m$, but due to the uniqueness of the Hamming code, the
generated code (or the orthogonal code) is the same.  Hence
$F(H_{F_1}) = \left(\begin{array}{c}H_m^{(F)}\\H_m^{(FF_1)}
\end{array}\right) = \left(\begin{array}{c}H_m\\H_m^{(F_2\lambda)} \end{array}\right)= \left(\begin{array}{c}H_m\\H_m^{(F_2)}
\end{array}\right)=H_{F_2}$. Therefore, the codes $\C_{F_1}$  and $\C_{F_2}$ defined in
(\ref{hf}) are equivalent.
\end{proof}

\begin{corollary}\label{CCZ-DesEqu}
Any two isomorphic self-embedding  permutations for the \\$STS(\cH^n)$
 are \textrm{CCZ}-equivalent.
\end{corollary}

\begin{proof}
Let $F_1$ and $F_2$ be two self-embedding permutations for the same $STS(\cH^n)$.
By Proposition \ref{SE-Codes}, the codes $\C_{F_1}$  and $\C_{F_2}$ are
equivalent, so the extended codes are equivalent, too.
Then, we conclude that $F_1, F_2$ are \textrm{CCZ}-equivalent.
\end{proof}

By Corollary \ref{CCZ-DesEqu}, it is possible to use the classification given by the self-embedding isomorphism,
in order to obtain a classification given by the \textrm{CCZ}-equivalence.
Note that the inverse of this result is not true in general. For example, for $m=4$, the permutations $\pi_{F_1}=(1, 15)(2, 3)(4, 5)(6,7)(9,10)(11,\allowbreak12)\allowbreak(13, 14)$ and $\pi_{F_2}=(1, 15)(2, 9)(3, 10)(4, 11)(5, 12)(6, 13)(7, 14)$
are \CCZ-e\-qui\-va\-lent \cite{RSV-acct, RSV}, but they do not define two isomorphic self-embedding permutations,
since they have 6 and 14 pinch points, respectively.
However, we can establish a weaker result considering just monomial power permutations, given by the next proposition.

\begin{proposition}
Let $F_1$, $F_2$ be two \textrm{CCZ}-equivalent monomial power permutations.
Then, $V^*_{F_1} =V^*_{F_2}$ and $v_{F_1} =v_{F_2}$.
\end{proposition}

\begin{proof}
Given a monomial power permutation $F$, by Proposition \ref{prop:monomialPower_Permutation},
we have that $v_F=v_F(a)$ and $V^*_F=V^*_F(a)$ for any $a\in \F^m\backslash \{\zero\}$,
so we can just take $a=1$. Moreover, we know that $v_{F}(1)$ is the number of values $b\not=\zero$ for which (\ref{equationAPN}) has solutions in $\F^m$ for $a=1$. Note that if $x$ is a solution of this equation, then
$x+b$ is also a solution, so the solutions come in pairs and the maximum value of $v_F$ is $2^{m-1}$, which is reached when $F$ is an \APN permutation, by Proposition \ref{SE_APN}. Hence, the permutations $F$ in all classes of
\textrm{CCZ}-equivalent monomial power permutations which are \APN satisfy that
$v_F=2^{m-1}$ and $V^*_F=\{ 1^\wedge(2^{m-1}-1)\}$. Note that if $F_1$ is an \APN function and $F_2$ is \textrm{CCZ}-equivalent to $F_2$, then $F_2$ is also an \APN function \cite{CCZ}.

If $F$ is not an \APN permutation, then~(\ref{equationAPN}) has more than a pair of solutions for some values of $b$ and $a=1$. When this happens, there is a connection with the quadruples in $\C^*_F$,
which is the extended code of $\C_F$. For example, if $x, x+b$ and $y, y+b$ are two different pairs of solutions of (\ref{equationAPN}), then the codeword given by the quadruple $(x,x+b,y,y+b)$
belongs to $\C^*_F$, since $F(x)+F(x+b)=F(y)+F(y+b)=a=1$.
Or, for instance, if $x,x+b$; $y,y+b$; and $z,z+b$ are three different pairs of solutions,
the quadruples $(x,x+b,y,y+b)$, $(x,x+b,z,z+b)$, $(y,y+b,z,z+b)$ give three codewords in $\C^*_{F}$ for the same argument.
In general, if $n_{b}$ is the number of solutions of (\ref{equationAPN}) for $b$ and $a=1$,
there are  $c_{b}=\binom{n_{b}/2}{2}$ quadruples in $\C^*_F$ associated to $b$.
Note that if $n_{b}=2$,  there is only a pair of solutions of (\ref{equationAPN}) and we have that $c_{b}=0$. Since the same quadruple is associated to $\frac{1}{2}\binom{4}{2}=3$ different values of $b$, the total number of quadruples in $\C^*_{F}$ is
 \begin{equation}\label{quad}
 \frac{2^m-1}{3}\sum_{b\in\F^m \backslash \{\zero \}} c_{b}.
 \end{equation}
Let $V^*_F=\{1^\wedge v^{(F)}_1, 2^\wedge v^{(F)}_2,\ldots\}$, where $v^{(F)}_i$ is the number of elements that appear $i$ times in
the multiset $\tilde{V}_F$. Then, $v_F-1=\sum_{i\geq 1} v^{(F)}_i$. The sum in~(\ref{quad}) has $\sum_{i>1} v^{(F)}_i$ nonzero terms corresponding to the values $b$ for which (\ref{equationAPN}) has more than a pair of solutions.

Let $F_1$ and $F_2$ be two \textrm{CCZ}-equivalent monomial power permutations, such that they are not \textrm{APN}.
Then, there exists a bijection between the codewords corresponding to the quadruples in both codes $\C^*_{F_1}$ and $\C^*_{F_2}$, given by a permutation $\pi$.
Hence, the set of $c_{b}$ quadruples in $\C^*_{F_1}$ goes to a set of $c_{{\bar b}}$ quadruples in $\C^*_{F_2}$ for an appropriate ${\bar b}$. Then, the number of quadruples in both codes $\C^*_{F_1}$ and $\C^*_{F_2}$ is the same:
$$\frac{2^m-1}{3}\sum_{b\in\F^m \backslash \{\zero \}} c_{b} = \frac{2^m-1}{3}\sum_{{\bar b}\in\F^m \backslash \{\zero \}} c_{{\bar b}},$$
where for each $b\in\F^m \backslash \{\zero \}$ there exists an appropriate ${\bar b}\in\F^m \backslash \{\zero \}$ such that $c_{b}=c_{{\bar b}}$. Not only the number of nonzero terms in the corresponding sums are the same, but also the repeated values. Hence, $v^{(F_1)}_i=v^{(F_2)}_i$ for $i>1$.
Moreover, since $\sum_{i\geq 1} i{\cdot }v^{(F_1)}_i=\sum_{i\ge 1} i{\cdot }v^{(F_2)}_i = 2^{m-1}-1$ and  $v^{(F_1)}_i=v^{(F_2)}_i$ for $i>1$, we can extend the equality for $i=1$. Therefore, we can conclude that $V^*_{F_1} =V^*_{F_2}$ and $v_{F_1} =1+\sum_{i\geq 1} v^{(F_1)}_i=1+\sum_{i\geq 1} v^{(F_1)}_i=v_{F_2}$.
 \end{proof}

It is clear that dealing with monomial power permutations, the rotation
lines at any two points
are the same up to a permutation. Therefore, it is enough to consider the rotation lines at one point,
for example, the point 1. The rotation line spectrum at point 1 can be
used to classify monomial power permutations, up to \se isomorphism, since any two isomorphic \se  permutations (regardless of they are monomial or not) have equivalent rotation schemes, so also the same rotation line spectrums up to a permutation.

For any $m\leq 17$,  Tables \ref{allPowerPermM13} and \ref{allPowerPermM15-17} show all \APN monomial
power permutations $F(x)=x^t$ over $\F^m$, taking just one representative up to
self-embedding isomorphism by Proposition \ref{MonomialPowerPerm}. For each class, the tables include the following information: the cyclotomic coset $C^*_t$, where the exponent $t$ belongs, the number of rotation lines $rl(1)$ at point $1$, and a reduced rotation line spectrum at point $1$.
For lack of space, the full rotation line spectrum is not given in these tables.
However, we describe a reduced rotation line spectrum including only the different
cardinalities of all rotation lines at point $1$,
since this is enough to distinguish all the cyclotomic classes $C_t^*$, which represent
all the nonisomorphic classes of \APN monomial power permutations.
Note that at least for all $m\leq 17$, all \APN monomial
power permutations in the same \textrm{CCZ}-equivalent class have
the same number of rotation lines, so the classification given by
the self-embedding isomorphism coincides with  the \textrm{CCZ}-equivalence.

\begin{proposition}\label{Melas}
Let $F(x)=x^{-1}$, so $\C_F$ is the Melas code. If $m$ is odd,
then in each point there are $(2^m-2)/6$ rotation lines with 6
points each. If $m$ is even, then in each point there are
$(2^m-4)/6$ rotation lines with 6 points each, and one rotation
line with 2 points.
\end{proposition}
\begin{proof}
Note that $F^{-1}(x)=F(x)$, since $F^2(x)=x$ for all $x\in
\F^m\backslash \{\zero\}$. Without loss of generality, we can
consider any point $a$ as the starting point. By the arguments shown after
Proposition \ref{def:SE}, the rotation lines
$R_1, R_2,\ldots, R_s$ at point $a$,  where $s=rl(a)$,
give a partition of the $n-1=2^m-2$ elements in $\F^m \backslash \{\zero,a\}$.
Given any of the rotation lines, $R_j$, we can write it as
$$R_j=[a_{1},a+a_{1}; a_{2},a+a_{2}; \ldots; a_{r_j}, a + a_{r_j}],$$
where $a_1$ is any element in $\F^m \backslash \{ \zero \}$ such
that  $a_1 \not = a$ and $a_{i+1}=F(F^{-1}(a)+F^{-1}(a+a_i))$ for
all $i\in\{1,\ldots, r_j-1\}$. It is easy to check that
$a_2=a_1$ if and only if $x^2+x+1=0$ has solutions over $GF(2^m)$,
so if and only if $m$ is even. When $m$ is even, the equation has
two solutions and we obtain a rotation line with 2 points.
Otherwise, when $a_1\not = a_2$, then $a_4=a_1$, and the rest of
rotation lines have always 6 points.
\end{proof}

From Tables \ref{allPowerPermM13} and \ref{allPowerPermM15-17}, it can be observed that
the minimum number of points in a rotation line is $6$.
In the next proposition, we prove that this is true in general for any $m$ and any \textrm{APN} permutation.

\begin{proposition} \label{upperlowerbound} Let $F$ be any bijective function over $\F^m$ such
that $F(\zero) = \zero$. If the permutation $F$ is \textrm{APN},
then any rotation line at any point has at least 6 points and at most $2^m-2$ points.
\end{proposition}
\begin{proof}
Note that the minimum distance of the code $\C_F={\cH}^n\cap \pi_F({\cH}^n)$ corresponding to an \textrm{APN} permutation $F$ is 5 \cite{CCZ}. Therefore, there is not any rotation line having 2 points, because there are no common triples in the Hamming codes ${\cH}^n$ and  $\pi_F({\cH}^n)$. Let us assume that there is a rotation line having 4 points for some element $a$: $[a_1,a+a_1;a_2,a+a_2]$. Then, the triples $(a,a_1,a+a_1),\, (a,a_2,a+a_2)$ belong to ${\cH}^n$ and the triples $(a,a+a_1,a_2),\, (a,a+a_2,a_1)$ belong to $\pi_F({\cH}^n).$ Since ${\cH}^n$ and  $\pi_F({\cH}^n)$ are linear codes, we obtain the common quadruple $(a_1,a+a_1,a_2,a+a_2)\in {\cH}^n\cap \pi_F({\cH}^n)=\C_F$.  Therefore, the minimum distance in $\C_F$ would be 4, which is a contradiction. Then, any rotation line at any point has at least 6 points.

The upper bound corresponds to the case when there is only one rotation line at a given point, so it has $2^m-2$ points.
\end{proof}

The lower bound given in Proposition \ref{upperlowerbound} is attainable  by the \APN permutation $F$ corresponding to the Melas code $\C_F$ for any length  $n=2^m-1$, where $m$ is odd, by Proposition \ref{Melas}. On the other hand, the upper bound corresponds to an \APN self-embedding permutation in a closed surface. These self-embeddings are pointed out in Table \ref{taulavb1} with $({\cdot })^{APN}$. Recall that they exist at least for $m\in \{3,5,7,11,17\}$, and there are none, at least for any non prime $m$
and for $m\in \{13,19\}$.

\section{Conclusions}
\label{sec:5}

We classified, up to isomorphism, all self-embedding monomial power permutations in close surfaces of the
Hamming Steiner triple system
$STS(\cH^n)$ for $m\leq 22$.
The existence of such self-embeddings and their classification for all prime $m\ge 23$ is still an open problem.
The found and classified ones are cyclic and nonorientable. The cyclicity
is proven
 for all $m$, and the nonorientability is checked only for all $m\leq 19$ using {\sc Magma}.
For $m\in \{3,5,7,11,17\}$, there exists one class of these permutations which is also \textrm{APN}, but for
$m\in \{13,19 \}$,
there is not any $\APN$ monomial power self-embedding permutations in a closed surface.

We established new invariants, $v_F$ and $V^*_F$, to distinguish CCZ-equivalent monomial power permutations.
Up to $m\leq 17$, the classification of \APN monomial power permutations, given by the self-embedding isomorphism, coincides with the
$\CCZ$-equivalence. It is still not known whether this is also true for any $m\geq 19$.
In any case, since two isomorphic self-embedding permutations are $\CCZ$-equivalent, we can use
the rotation line spectrum as a first step to obtain
a classification, up to $\CCZ$-equivalence, for any permutation not
only for monomial power permutations.

\newpage

{\Large {\bf Appendix}}
\begin{table}[htp!]
\begin{center}
{\footnotesize
$$
\begin{array}{|c|c|c|c|c|c|c|} \hline
m&C^*_t & v_F \\ \hline\hline 3 & C^*_{3} & 4^{\mbox{\tiny
APN}} \\ \hline 5 & C^*_{5} & 16^{\mbox{\tiny APN}} \\ \hline 7 &
C^*_{19} & 43 \\7 & C^*_{ 7 }, C^*_{ 21 } & 50 \\7 & C^*_{ 9 }
& 64 ^{\mbox{\tiny APN}} \\
\hline 11 & C^*_{ 39 } & 683 \\11 & C^*_{ 59 } & 694 \\11 & C^*_{
371 } & 738 \\11 & C^*_{ 181 } & 760 \\11 & C^*_{ 37 } & 771 \\11
& C^*_{ 25 } & 793 \\11 & C^*_{ 21 }, C^*_{ 687 } & 815 \\11 &
C^*_{ 73 }, C^*_{ 165 } & 826 \\11 & C^*_{ 101 } & 837 \\11 &
C^*_{ 127 } & 870 \\11 & C^*_{ 317 } & 881 \\11 & C^*_{ 107 } &
1024^{\mbox{\tiny APN}} \\\hline 13 & C^*_{ 51 } & 2887 \\13 &
C^*_{ 587 } & 3004 \\13 & C^*_{ 659 } & 3108 \\13 & C^*_{ 295 } &
3186 \\13 & C^*_{ 249 }, C^*_{ 661 } & 3199 \\13 & C^*_{ 75 } &
3251 \\13 & C^*_{ 151 } & 3316 \\13 & C^*_{ 133 }, C^*_{ 605 } &
3342 \\13 & C^*_{ 875 } & 3381 \\13 & C^*_{ 93 } & 3407\\ \hline
17 & C^*_{ 6827 } & 50457 \\17 & C^*_{ 13803 } & 50610 \\17 &
C^*_{ 5451 } & 50661 \\17 & C^*_{ 6059 } & 50678 \\17 & C^*_{ 1129
} & 50712 \\17 & C^*_{ 15691 } & 50746 \\17 & C^*_{ 8081 } & 50814
\\17 & C^*_{ 4457 }, C^*_{ 24285 } & 50865 \\17 & C^*_{ 2185 } &
50933 \\17 & C^*_{ 2387 }, C^*_{ 6705 } & 50950 \\17 & C^*_{ 1223
} & 51001 \\17 & C^*_{ 5681 } & 51018  \\ \hline\hline
\end{array}
\hspace*{1cm}
\begin{array}{|c|c|c|c|c|c|c|} \hline
m&C^*_t & v_F \\ \hline\hline 17 & C^*_{ 23987 } & 51069 \\17 &
C^*_{ 2043 } & 51154 \\17 & C^*_{ 4533 } & 51171 \\17 & C^*_{
10171 } & 51273 \\17 & C^*_{ 5003 } & 51307 \\17 & C^*_{ 249 } &
51324 \\17 & C^*_{ 2363 }, C^*_{ 11071 } & 51341 \\17 & C^*_{ 1673
}, C^*_{ 9909 } & 51358 \\17 & C^*_{ 3163 }, C^*_{ 8917 } & 51409
\\17 & C^*_{ 5965 } & 51460 \\17 & C^*_{ 11955 } & 51494 \\17 &
C^*_{ 2335 } & 51511 \\17 & C^*_{ 285 } & 51528 \\17 & C^*_{ 4285
} & 51579 \\17 & C^*_{ 4233 } & 51596 \\17 & C^*_{ 3689 }, C^*_{
4743 } & 51613 \\17 & C^*_{ 421 } & 51630 \\17 & C^*_{ 543 },
C^*_{ 7143 } & 51647 \\17 & C^*_{ 1791 }, C^*_{ 4931 }, C^*_{ 5947
} & 51715 \\17 & C^*_{ 2851 }, C^*_{ 4519 } & 51749 \\17 & C^*_{
1201 }, C^*_{ 1949 } & 51783 \\17 & C^*_{ 2621 } & 51851 \\17 &
C^*_{ 1517 } & 51868 \\17 & C^*_{ 4635 }, C^*_{ 5663 } & 51936
\\17 & C^*_{ 1891 } & 51953 \\17 & C^*_{ 1313 } & 52021 \\17 &
C^*_{ 1395 } & 52038 \\17 & C^*_{ 137 }, C^*_{ 3309 } & 52089 \\17
& C^*_{ 1001 }, C^*_{ 2979 } & 52123 \\17 & C^*_{ 3757 } & 52157
\\17 & C^*_{ 6967 } & 52480 \\17 & C^*_{ 6249 } & 52531 \\17 &
C^*_{ 1431 } & 52548 \\17 & C^*_{ 2673 } & 52599 \\17 & C^*_{ 151
} & 52616 \\17 & C^*_{ 2281 } & 52633 \\17 & C^*_{ 907 } & 52837
\\17 & C^*_{ 4499 } & 53279 \\17 & C^*_{ 257 } &
65536^{\mbox{\tiny APN}} \\  \hline\hline
\end{array}
$$}
\end{center}
\caption{\label{taulavb1} Classification of all self-embedding monomial power
permutations in closed surfaces, $F(x)=x^t$ over $\F^m$, based on the invariant $v_F$ for $m\in
\{3,\ldots,18\}$.}
\end{table}

\newpage
\begin{table}[htp!]
\begin{center}
{\footnotesize
$$
\begin{array}{|c|c|c|c|c|c|c|} \hline
m&C^*_t & v_F \\ \hline\hline 19 & C^*_{ 12895 } & 201439 \\19
& C^*_{ 7989 } & 203491 \\19 & C^*_{ 13643 } & 204365 \\19 & C^*_{
21847 } & 204631 \\19 & C^*_{ 28565 } & 204669 \\19 & C^*_{ 38283
} & 204688 \\19 & C^*_{ 50799 } & 204707 \\19 & C^*_{ 7021 } &
204726 \\19 & C^*_{ 1257 } & 204745 \\19 & C^*_{ 15003 } & 204821
\\19 & C^*_{ 6501 }, C^*_{ 37561 }, C^*_{ 59999 } & 204840 \\19 &
C^*_{ 19367 } & 204935 \\19 & C^*_{ 35373 } & 205106 \\19 & C^*_{
24533 } & 205125 \\19 & C^*_{ 24041 } & 205296 \\19 & C^*_{ 80573
} & 205334 \\19 & C^*_{ 15593 } & 205353 \\19 & C^*_{ 8487 },
C^*_{ 38045 } & 205372 \\19 & C^*_{ 28495 }, C^*_{ 64441 } &
205429 \\19 & C^*_{ 4779 } & 205524 \\19 & C^*_{ 16077 } & 205543
\\19 & C^*_{ 12661 }, C^*_{ 26441 } & 205638 \\19 & C^*_{ 4277 },
C^*_{ 23311 } & 205676 \\19 & C^*_{ 27385 } & 205733 \\19 & C^*_{
14699 } & 205809 \\19 & C^*_{ 4729 } & 205828 \\19 & C^*_{ 877 } &
205942 \\19 & C^*_{ 1211 }, C^*_{ 8871 } & 205961 \\19 & C^*_{
7011 } & 205980 \\19 & C^*_{ 62651 } & 206018 \\19 & C^*_{ 15449
}, C^*_{ 56575 } & 206075 \\19 & C^*_{ 38891 } & 206113 \\19 &
C^*_{ 10475 } & 206132 \\19 & C^*_{ 12213 }, C^*_{ 18267 } &
206227 \\19 & C^*_{ 54003 } & 206265 \\19 & C^*_{ 35571 } & 206284
\\ \hline\hline
\end{array}
\hspace*{1cm}
\begin{array}{|c|c|c|c|c|c|c|} \hline
m&C^*_t & v_F \\ \hline\hline 19 & C^*_{ 2391 }, C^*_{ 26219 },
C^*_{ 32479 } & 206341 \\19 & C^*_{ 2987 }, C^*_{ 31923 } & 206398
\\19 & C^*_{ 42579 } & 206436 \\19 & C^*_{ 22475 } & 206512 \\19 &
C^*_{ 11513 } & 206531 \\19 & C^*_{ 11039 }, C^*_{ 12447 } &
206550 \\19 & C^*_{ 1531 }, C^*_{ 46503 } & 206569 \\19 & C^*_{
7147 } & 206607 \\19 & C^*_{ 13127 }, C^*_{ 17629 } & 206797 \\19
& C^*_{ 13225 } & 206816 \\19 & C^*_{ 10633 } & 206835 \\19 &
C^*_{ 42213 } & 206949 \\19 & C^*_{ 235 } & 206968 \\19 & C^*_{
28461 } & 207006 \\19 & C^*_{ 1275 } & 207025 \\19 & C^*_{ 30917 }
& 207158 \\19 & C^*_{ 4665 } & 207177 \\19 & C^*_{ 32359 }, C^*_{
62927 } & 207234 \\19 & C^*_{ 7769 } & 207253 \\19 & C^*_{ 48967 }
& 207291 \\19 & C^*_{ 58295 } & 207310 \\19 & C^*_{ 16949 } &
207405 \\19 & C^*_{ 38521 } & 207804 \\19 & C^*_{ 9515 } & 207861
\\19 & C^*_{ 9539 } & 207880 \\19 & C^*_{ 30677 } & 207956 \\19 &
C^*_{ 4369 } & 208013 \\19 & C^*_{ 47463 } & 208051 \\19 & C^*_{
3337 } & 208070 \\19 & C^*_{ 5057 } & 208146 \\19 & C^*_{ 7241 } &
208241 \\19 & C^*_{ 23803 } & 208355 \\19 & C^*_{ 9785 } & 208678
\\19 & C^*_{ 503 } & 208925 \\19 & C^*_{ 20657 } & 208963 \\19 &
C^*_{ 28201 } & 209837 \\ \hline\hline
\end{array}
$$}
\caption{\label{taulavb2} Classification of all self-embedding monomial power
permutations in  closed surfaces, $F(x)=x^t$ over $F^m$, based on the invariant $v_F$ for $m=19$.}
\end{center}
\end{table}

\begin{table}[htp!]
\begin{center}
{\footnotesize
$$
\begin{array}{|c|c|c|l|} \hline
m&C^*_t & v_F & \hspace*{3cm}V^*_F \\
\hline
7 & C^*_{ 7 } & 50 & \{ 1^\wedge{42}, 3^\wedge{7 } \}
\\
7 & C^*_{ 21 } & 50 & \{ 1^\wedge{42}, 3^\wedge{7 } \}
\\
\hline
11 & C^*_{ 21 } & 815 & \{ 1^\wedge{627}, 2^\wedge{165}, 3^\wedge{22 } \}
\\
11 & C^*_{ 687 } & 815 & \{ 1^\wedge{627}, 2^\wedge{165}, 3^\wedge{22 } \}
\\
11 & C^*_{ 73 } & 826 & \{ 1^\wedge{660}, 2^\wedge{132}, 3^\wedge{33 } \}
\\
11 & C^*_{ 165 } & 826 & \{ 1^\wedge{682}, 2^\wedge{99}, 3^\wedge{33}, 4^\wedge{11 } \}
\\
\hline
13 & C^*_{ 249 } & 3199 & \{ 1^\wedge{2444}, 2^\wedge{624}, 3^\wedge{117}, 4^\wedge{13 } \}
\\
13 & C^*_{ 661 } & 3199 & \{ 1^\wedge{2470}, 2^\wedge{585}, 3^\wedge{117}, 4^\wedge{26 } \}
\\
13 & C^*_{ 133 } & 3342 & \{ 1^\wedge{2691}, 2^\wedge{546}, 3^\wedge{104 } \}
\\
13 & C^*_{ 605 } & 3342 & \{ 1^\wedge{2678}, 2^\wedge{585}, 3^\wedge{65}, 4^\wedge{13 } \}
\\
\hline
17 & C^*_{ 4457 } & 50865 & \{ 1^\wedge{38675}, 2^\wedge{10132}, 3^\wedge{1700}, 4^\wedge{289}, 5^\wedge{68 } \}
\\
17 & C^*_{ 24285 } & 50865 & \{ 1^\wedge{38692}, 2^\wedge{10013}, 3^\wedge{1853}, 4^\wedge{272}, 5^\wedge{34 } \}
\\
17 & C^*_{ 2387 } & 50950 & \{ 1^\wedge{38692}, 2^\wedge{10183}, 3^\wedge{1836}, 4^\wedge{221}, 5^\wedge{17 } \}
\\
17 & C^*_{ 6705 } & 50950 & \{ 1^\wedge{38352}, 2^\wedge{10829}, 3^\wedge{1581}, 4^\wedge{153}, 5^\wedge{34 } \}
\\
17 & C^*_{ 2363 } & 51341 & \{ 1^\wedge{39202}, 2^\wedge{10234}, 3^\wedge{1768}, 4^\wedge{119}, 5^\wedge{17 } \}
\\
17 & C^*_{ 11071 } & 51341 & \{ 1^\wedge{39304}, 2^\wedge{10166}, 3^\wedge{1615}, 4^\wedge{221}, 5^\wedge{34 } \}
\\
17 & C^*_{ 1673 } & 51358 & \{ 1^\wedge{39304}, 2^\wedge{10149}, 3^\wedge{1683}, 4^\wedge{221 } \}
\\
17 & C^*_{ 9909 } & 51358 & \{ 1^\wedge{39576}, 2^\wedge{9707}, 3^\wedge{1785}, 4^\wedge{255}, 5^\wedge{34 } \}
\\
17 & C^*_{ 3163 } & 51409 & \{ 1^\wedge{39168}, 2^\wedge{10591}, 3^\wedge{1462}, 4^\wedge{153}, 5^\wedge{17}, 6^\wedge{17
} \}
\\
17 & C^*_{ 8917 } & 51409 & \{ 1^\wedge{39423}, 2^\wedge{10149}, 3^\wedge{1581}, 4^\wedge{204}, 5^\wedge{51 } \}
\\
17 & C^*_{ 3689 } & 51613 & \{ 1^\wedge{40171}, 2^\wedge{9265}, 3^\wedge{1921}, 4^\wedge{221}, 5^\wedge{17}, 6^\wedge{17
} \}
\\
17 & C^*_{ 4743 } & 51613 & \{ 1^\wedge{39933}, 2^\wedge{9656}, 3^\wedge{1836}, 4^\wedge{153}, 5^\wedge{34 } \}
\\
17 & C^*_{ 543 } & 51647 & \{ 1^\wedge{39848}, 2^\wedge{9996}, 3^\wedge{1564}, 4^\wedge{187}, 5^\wedge{51 } \}
\\
17 & C^*_{ 7143 } & 51647 & \{ 1^\wedge{39848}, 2^\wedge{9877}, 3^\wedge{1768}, 4^\wedge{136}, 5^\wedge{17 } \}
\\
17 & C^*_{ 1791 } & 51715 & \{ 1^\wedge{39882}, 2^\wedge{10030}, 3^\wedge{1632}, 4^\wedge{153}, 5^\wedge{17 } \}
\\
17 & C^*_{ 4931 } & 51715 & \{ 1^\wedge{39916}, 2^\wedge{9962}, 3^\wedge{1683}, 4^\wedge{119}, 5^\wedge{34 } \}
\\
17 & C^*_{ 5947 } & 51715 & \{ 1^\wedge{39848}, 2^\wedge{10149}, 3^\wedge{1479}, 4^\wedge{238 } \}
\\
17 & C^*_{ 2851 } & 51749 & \{ 1^\wedge{39916}, 2^\wedge{10064}, 3^\wedge{1615}, 4^\wedge{119}, 5^\wedge{34 } \}
\\
17 & C^*_{ 4519 } & 51749 & \{ 1^\wedge{40290}, 2^\wedge{9367}, 3^\wedge{1870}, 4^\wedge{204}, 5^\wedge{17 } \}
\\
17 & C^*_{ 1201 } & 51783 & \{ 1^\wedge{40052}, 2^\wedge{9945}, 3^\wedge{1598}, 4^\wedge{136}, 5^\wedge{51 } \}
\\
17 & C^*_{ 1949 } & 51783 & \{ 1^\wedge{40307}, 2^\wedge{9520}, 3^\wedge{1700}, 4^\wedge{187}, 5^\wedge{68 } \}
\\
17 & C^*_{ 4635 } & 51936 & \{ 1^\wedge{40358}, 2^\wedge{9690}, 3^\wedge{1751}, 4^\wedge{136 } \}
\\
17 & C^*_{ 5663 } & 51936 & \{ 1^\wedge{40392}, 2^\wedge{9724}, 3^\wedge{1598}, 4^\wedge{204}, 5^\wedge{17 } \}
\\
17 & C^*_{ 137 } & 52089 & \{ 1^\wedge{40562}, 2^\wedge{9962}, 3^\wedge{1292}, 4^\wedge{187}, 5^\wedge{85 } \}
\\
17 & C^*_{ 3309 } & 52089 & \{ 1^\wedge{40596}, 2^\wedge{9741}, 3^\wedge{1547}, 4^\wedge{204 } \}
\\
17 & C^*_{ 1001 } & 52123 & \{ 1^\wedge{40851}, 2^\wedge{9418}, 3^\wedge{1615}, 4^\wedge{204}, 5^\wedge{17}, 6^\wedge{17
} \}
\\
17 & C^*_{ 2979 } & 52123 & \{ 1^\wedge{40783}, 2^\wedge{9452}, 3^\wedge{1734}, 4^\wedge{119}, 5^\wedge{34 } \}
\\
\hline
\end{array}
$$}
\end{center}
\caption{\label{taula_multiplicities1} Classification of some self-embedding monomial power
permutations in closed surfaces, $F(x)=x^t$ over $\F^m$, based on the invariants $v_F$ and $V^*_F$ for $m\in \{7,11,13,17\}$.}
\end{table}

\begin{table}[htp!]
\begin{center}
{\footnotesize
$$
\begin{array}{|c|c|c|l|} \hline
m&C^*_t & v_F & \hspace*{3cm}V^*_F \\
\hline
19 & C^*_{ 6501 } & 204840 & \{ 1^\wedge{157377}, 2^\wedge{38779}, 3^\wedge{7676}, 4^\wedge{855}, 5^\wedge{152 } \}
\\
19 & C^*_{ 37561 } & 204840 & \{ 1^\wedge{156522}, 2^\wedge{40527}, 3^\wedge{6764}, 4^\wedge{893}, 5^\wedge{114},
7^\wedge{19 } \}
\\
19 & C^*_{ 59999 } & 204840 & \{ 1^\wedge{156503}, 2^\wedge{40565}, 3^\wedge{6764}, 4^\wedge{817}, 5^\wedge{190 } \}
\\
19 & C^*_{ 8487 } & 205372 & \{ 1^\wedge{157206}, 2^\wedge{40584}, 3^\wedge{6726}, 4^\wedge{722}, 5^\wedge{114},
7^\wedge{19 } \}
\\
19 & C^*_{ 38045 } & 205372 & \{ 1^\wedge{157567}, 2^\wedge{39748}, 3^\wedge{7182}, 4^\wedge{836}, 5^\wedge{38 } \}
\\
19 & C^*_{ 28495 } & 205429 & \{ 1^\wedge{157719}, 2^\wedge{39881}, 3^\wedge{6783}, 4^\wedge{912}, 5^\wedge{133 } \}
\\
19 & C^*_{ 64441 } & 205429 & \{ 1^\wedge{157529}, 2^\wedge{40318}, 3^\wedge{6479}, 4^\wedge{988}, 5^\wedge{95},
6^\wedge{19 } \}
\\
19 & C^*_{ 12661 } & 205638 & \{ 1^\wedge{158669}, 2^\wedge{38969}, 3^\wedge{6745}, 4^\wedge{988}, 5^\wedge{247},
6^\wedge{19 } \}
\\
19 & C^*_{ 26441 } & 205638 & \{ 1^\wedge{157833}, 2^\wedge{40147}, 3^\wedge{6707}, 4^\wedge{855}, 5^\wedge{95 } \}
\\
19 & C^*_{ 4277 } & 205676 & \{ 1^\wedge{157814}, 2^\wedge{40109}, 3^\wedge{7011}, 4^\wedge{646}, 5^\wedge{76},
6^\wedge{19 } \}
\\
19 & C^*_{ 23311 } & 205676 & \{ 1^\wedge{158042}, 2^\wedge{39881}, 3^\wedge{6745}, 4^\wedge{931}, 5^\wedge{76 } \}
\\
19 & C^*_{ 1211 } & 205961 & \{ 1^\wedge{158004}, 2^\wedge{40831}, 3^\wedge{6061}, 4^\wedge{1026}, 5^\wedge{38 } \}
\\
19 & C^*_{ 8871 } & 205961 & \{ 1^\wedge{158327}, 2^\wedge{40242}, 3^\wedge{6346}, 4^\wedge{931}, 5^\wedge{114 } \}
\\
19 & C^*_{ 15449 } & 206075 & \{ 1^\wedge{158612}, 2^\wedge{39843}, 3^\wedge{6745}, 4^\wedge{760}, 5^\wedge{114 } \}
\\
19 & C^*_{ 56575 } & 206075 & \{ 1^\wedge{158916}, 2^\wedge{39444}, 3^\wedge{6574}, 4^\wedge{1083}, 5^\wedge{57 } \}
\\
19 & C^*_{ 12213 } & 206227 & \{ 1^\wedge{158878}, 2^\wedge{39653}, 3^\wedge{6840}, 4^\wedge{836}, 5^\wedge{19 } \}
\\
19 & C^*_{ 18267 } & 206227 & \{ 1^\wedge{158840}, 2^\wedge{39862}, 3^\wedge{6707}, 4^\wedge{646}, 5^\wedge{152},
6^\wedge{19 } \}
\\
19 & C^*_{ 2391 } & 206341 & \{ 1^\wedge{159315}, 2^\wedge{39216}, 3^\wedge{6954}, 4^\wedge{741}, 5^\wedge{114 } \}
\\
19 & C^*_{ 26219 } & 206341 & \{ 1^\wedge{158764}, 2^\wedge{40470}, 3^\wedge{6080}, 4^\wedge{950}, 5^\wedge{57},
6^\wedge{19 } \}
\\
19 & C^*_{ 32479 } & 206341 & \{ 1^\wedge{158726}, 2^\wedge{40147}, 3^\wedge{6783}, 4^\wedge{646}, 5^\wedge{38 } \}
\\
19 & C^*_{ 2987 } & 206398 & \{ 1^\wedge{158479}, 2^\wedge{40793}, 3^\wedge{6498}, 4^\wedge{551}, 5^\wedge{76 } \}
\\
19 & C^*_{ 31923 } & 206398 & \{ 1^\wedge{159011}, 2^\wedge{39881}, 3^\wedge{6669}, 4^\wedge{817}, 5^\wedge{19 } \}
\\
19 & C^*_{ 11039 } & 206550 & \{ 1^\wedge{158707}, 2^\wedge{40983}, 3^\wedge{6023}, 4^\wedge{779}, 5^\wedge{57 } \}
\\
19 & C^*_{ 12447 } & 206550 & \{ 1^\wedge{159600}, 2^\wedge{39235}, 3^\wedge{6859}, 4^\wedge{798}, 5^\wedge{38},
6^\wedge{19 } \}
\\
19 & C^*_{ 1531 } & 206569 & \{ 1^\wedge{159011}, 2^\wedge{40641}, 3^\wedge{5909}, 4^\wedge{931}, 5^\wedge{57},
6^\wedge{19 } \}
\\
19 & C^*_{ 46503 } & 206569 & \{ 1^\wedge{159068}, 2^\wedge{40432}, 3^\wedge{6099}, 4^\wedge{931}, 5^\wedge{38 } \}
\\
19 & C^*_{ 13127 } & 206797 & \{ 1^\wedge{159752}, 2^\wedge{39919}, 3^\wedge{6099}, 4^\wedge{874}, 5^\wedge{152 } \}
\\
19 & C^*_{ 17629 } & 206797 & \{ 1^\wedge{159790}, 2^\wedge{39672}, 3^\wedge{6441}, 4^\wedge{779}, 5^\wedge{114 } \}
\\
19 & C^*_{ 32359 } & 207234 & \{ 1^\wedge{160037}, 2^\wedge{40413}, 3^\wedge{5947}, 4^\wedge{760}, 5^\wedge{57},
6^\wedge{19 } \}
\\
19 & C^*_{ 62927 } & 207234 & \{ 1^\wedge{160531}, 2^\wedge{39254}, 3^\wedge{6745}, 4^\wedge{665}, 5^\wedge{19},
6^\wedge{19 } \}
\\ \hline
\end{array}
$$}
\end{center}
\caption{\label{taula_multiplicities2} Classification of some self-embedding monomial power
permutations in closed surfaces, $F(x)=x^t$ over $\F^m$, based on the invariants $v_F$ and $V^*_F$ for $m=19$.}
\end{table}

\begin{table}[ht]
\begin{center}
{\footnotesize

$$
\begin{array}{|c|l|c|l|} \hline
m & C^*_t & rl(1) & \mbox{\hspace*{1.8cm}{reduced rotation line spectrum at point 1}}    \\ \hline
\hline
3 & C^*_{ 3 } & 1 & ( 1 ; 6 ) \\
\hline
5 & C^*_{ 5 } & 1 & ( 1 ; 30 ) \\
5 & C^*_{ 3 } & 2 & ( 2 ; 10, 20 ) \\
5 & C^*_{ 15 } & 5 & ( 5 ; 6 ) \\
\hline
7 & C^*_{ 9 } & 1 & ( 1 ; 126 ) \\
7 & C^*_{ 5 } & 2 & ( 2 ; 28, 98 ) \\
7 & C^*_{ 3 } & 4 & ( 4 ; 14, 28, 42 ) \\
7 & C^*_{ 23 } & 4 & ( 4 ; 14, 84 ) \\
7 & C^*_{ 11 } & 15 & ( 15 ; 6, 10, 14 ) \\
7 & C^*_{ 63 } & 21 & ( 21 ; 6 ) \\
\hline
9 & C^*_{ 47 } & 3 & ( 3 ; 6, 234, 270 ) \\
9 & C^*_{ 5 } & 5 & ( 5 ; 6, 120, 144 ) \\
9 & C^*_{ 13 } & 5 & ( 5 ; 6, 54, 126, 270 ) \\
9 & C^*_{ 17 } & 5 & ( 5 ; 6, 72, 144 ) \\
9 & C^*_{ 3 } & 10 & ( 10 ; 6, 24, 36, 54, 72, 90 ) \\
9 & C^*_{ 19 } & 14 & ( 14 ; 6, 18, 36, 40, 54 ) \\
9 & C^*_{ 255 } & 85 & ( 85 ; 6 ) \\
\hline
11 & C^*_{ 107 } & 1 & ( 1 ; 2046 ) \\
11 & C^*_{ 35 } & 3 & ( 3 ; 264, 682, 1100 ) \\
11 & C^*_{ 95 } & 4 & ( 4 ; 22, 374, 1276 ) \\
11 & C^*_{ 5 } & 6 & ( 6 ; 22, 88, 132, 396, 462, 946 ) \\
11 & C^*_{ 57 } & 6 & ( 6 ; 22, 44, 66, 88, 440, 1386 ) \\
11 & C^*_{ 17 } & 8 & ( 8 ; 22, 66, 110, 132, 154, 264, 396, 902 ) \\
11 & C^*_{ 9 } & 13 & ( 13 ; 22, 136, 528 ) \\
11 & C^*_{ 33 } & 13 & ( 13 ; 22, 88, 176 ) \\
11 & C^*_{ 13 } & 14 & ( 14 ; 88, 112, 176, 550 ) \\
11 & C^*_{ 3 } & 18 & ( 18 ; 22, 44, 66, 88, 110, 132, 154, 176, 198, 242 ) \\
11 & C^*_{ 43 } & 19 & ( 19 ; 18, 44, 66, 88, 308, 330, 396, 550 ) \\
11 & C^*_{ 1023 } & 341 & ( 341 ; 6 ) \\
\hline
13 & C^*_{ 71 } & 3 & ( 3 ; 312, 364, 7514 ) \\
13 & C^*_{ 9 } & 3 & ( 3 ; 26, 338, 7826 ) \\
13 & C^*_{ 67 } & 3 & ( 3 ; 104, 7982 ) \\
13 & C^*_{ 171 } & 3 & ( 3 ; 26, 2002, 6162 ) \\
13 & C^*_{ 5 } & 5 & ( 5 ; 156, 234, 338, 806, 6656 ) \\
13 & C^*_{ 287 } & 6 & ( 6 ; 26, 156, 390, 754, 2496, 4368 ) \\
13 & C^*_{ 33 } & 6 & ( 6 ; 26, 78, 338, 1196, 6474 ) \\
13 & C^*_{ 191 } & 6 & ( 6 ; 26, 52, 286, 3302, 4472 ) \\
13 & C^*_{ 17 } & 8 & ( 8 ; 26, 52, 78, 806, 1976, 2262, 2964 ) \\
13 & C^*_{ 13 } & 9 & ( 9 ; 26, 52, 78, 156, 624, 3536, 3666 ) \\
13 & C^*_{ 65 } & 13 & ( 13 ; 630 ) \\
13 & C^*_{ 57 } & 18 & ( 18 ; 16, 52, 130, 234, 260, 7306 ) \\
13 & C^*_{ 3 } & 52 & ( 52 ; 26, 32, 78, 104, 130, 156, 182, 208, 234, 260,
286, 312, 338, 364, 468 ) \\
13 & C^*_{ 4095 } & 1365 & ( 1365 ; 6 ) \\
\hline\hline
\end{array}
$$
}
\caption{\label{allPowerPermM13} Classification of all \APN monomial
power permutations for $m\leq 13$ using the invariant given by the
rotation line spectrum.}
\end{center}
\end{table}
\newpage

\begin{table}[ht]
\begin{center}
\footnotesize{
$$
\begin{array}{|c|l|c|l|} \hline
m & C^*_t & rl(1) & \mbox{\hspace*{1.8cm}{reduced rotation line spectrum at point 1}}    \\ \hline
\hline
15& C^*_{131}   &  10  &   (10;6, 30, 2170, 8720) \\
15   &  C^*_{241}   &   15  &  (15;6, 10, 20, 180, 380, 1500, 2330, 22860) \\
15& C^*_{13}   &   16  &   (16;6, 10, 20, 36, 90, 720,10560) \\
15& C^*_{1371}   &   16 &   (16;6, 30, 210, 288, 430, 750,1230, 4770, 5490,17550) \\
15& C^*_{383}  &   25  &   (25;6, 10, 20, 158, 180, 330, 900, 2530, 21360) \\
15& C^*_5   &   30  &   (30;6, 30, 70, 108, 208, 9600) \\
15& C^*_{17}   &   47  &   (47;6, 10, 20, 30, 32,306, 430, 1040, 2640, 2700, 3750, 11700) \\
15& C^*_{129}  &   117  &   (117;6, 30, 60, 150, 300) \\
15& C^*_3   &   260  &   (260; 6, 10, 20, 22, 36, 40,50, 60, 66, 70, 72, 78, 80,90, 110, 120, 130, 140, 150,  \\
&&&   160, 180,200, 210, 240, 250, 260, 300, 330, 350,360, 390, 420, 450, 480,  \\
&&&   510, 540, 570,600, 660, 690, 720, 750, 1020, 1050,1110) \\
15& C^*_{3657}   &   341  &   (341;6, 10, 12, 16, 18, 20, 22,82, 86, 184, 220,264, 278, 364, 384, 462) \\
  15   &  C^*_{16383}  & 5461 &  (5461;6 )  \\ \hline

  17   &  C^*_{257}   &   1  &   (1;131070) \\
 17   &  C^*_{65}   &   4  &   (4;170, 680, 4386, 125834) \\
  17   &  C^*_{271}   &   6  &   (6;34, 102, 306, 4420, 24684, 101524) \\
  17   &  C^*_{9}   &   9  &   (9;34, 102, 238,544, 850, 1632, 11798, 28186, 87686) \\
   17   &  C^*_{683}   &   9  &   (9;272, 748, 1156, 2720, 5746, 9656, 11288, 15878,83606) \\
  17   &  C^*_{1151}   &   10  &   (10;68, 714, 1224, 4522, 4828,4964, 6086, 45934,57902) \\
 17   &  C^*_{33}   &   12  &   (12;102, 238, 306, 476, 646, 1972,2550, 3298, 6018, 17850, 41956, 55658) \\
  17  &  C^*_{129}   &   12  &   (12;68, 102, 136, 272, 374, 1972, 15096,16082, 20876, 24174, 51816) \\
  17   &  C^*_{13}   &  14  &   (14;34, 68, 204, 714, 884, 1394, 2380, 6936, 12580, 12988,18904, 26486, 47464) \\
  17   &  C^*_{259}   &   21 &   (21;34, 204, 1122, 7628) \\
  17   &  C^*_{767}   &   24  &   (24;68, 510, 718, 1632, 3468, 5882, 29614, 77690) \\
  17   &  C^*_{5}   &   25  &   (25;34, 68, 306, 11152, 16830,35360, 66708) \\
  17   &  C^*_{57}   &  25  &   (25;68, 136, 646, 816, 850, 1156, 3318, 3332, 67660) \\
  17   &  C^*_{241}   &   26  &  (26;34, 90, 510, 748, 1700, 1734, 24582, 24684, 75514) \\
  17   &  C^*_{993}   &   40  &   (40;34, 128, 136, 190, 340, 8806, 116212) \\
  17   &  C^*_{17}   &   59  &   (59;34,160, 204, 442, 1384, 1462, 1632, 3022, 5542, 17272, 26860) \\
  17   & C^*_{3}   &   388  &   (388;34, 36, 44, 50, 54, 64, 68, 102, 104,136, 170,204, 238, 272, 306, 340, \\
&&&   408, 442, 476, 510, 544, 578, 612, 646, 680, 748, 782,816, 850, 884,  \\
&&&  918, 952, 1020, 1054, 1088,1156, 1190, 1224, 1292, 1326, 1360, \\
&&&  1394, 1462, 1496, 1564, 1598, 1768, 1904, 1972, 2006) \\
 17   &   C^*_{65535}  & 21845&  (21845;6) \\ \hline\hline
\end{array}
$$
}
\caption{\label{allPowerPermM15-17} Classification of all \APN monomial
power permutations for $m\in \{15,17\}$ using the invariant given
by the rotation line spectrum.}
\end{center}
\end{table}

\end{document}